\documentclass[conference]{IEEEtran}

\makeatletter

\def\ps@headings{%

\def\@evenhead{\scriptsize\thepage \hfil \leftmark\mbox{}}%

\def\@oddfoot{}%

\def\@evenfoot{}}

\makeatother

\usepackage{amssymb}

\usepackage{algorithmic}

\usepackage{algorithm}
\usepackage{verbatim}
\usepackage{amsmath}
\usepackage{enumerate}
\usepackage{fancybox}

\usepackage{verbatim}
\usepackage{array}
\usepackage{graphicx}
\usepackage{float}%\usepackage[tight,footnotesize]{subfigure}
\usepackage[footnotesize]{subfigure}
{\begin{Sbox}\begin{minipage}}%
{\end{minipage}\end{Sbox}\fbox{\TheSbox}}

\newtheorem{theorem}{Theorem}
\begin{document}

\title{Towards  Loop-Free 
Forwarding of  Anonymous Internet Datagrams that Enforce Provenance }

\author{J.J. Garcia-Luna-Aceves$^{1,2}$ \\
$^1$Department of Computer Engineering,
 University of California, Santa Cruz, CA 95064\\
$^2$Palo Alto Research Center, Palo Alto, CA 94304 \\
Email: jj@soe.ucsc.edu }

\maketitle

\begin{abstract}

The  way in which addressing and forwarding are  implemented in the Internet  constitutes one of its biggest  privacy and security challenges. The fact that  source addresses in Internet datagrams cannot be trusted makes the IP Internet inherently vulnerable to  DoS and DDoS attacks. The Internet forwarding plane is open to attacks to the privacy of datagram sources, because  source addresses in Internet datagrams have global scope. 
The fact an Internet datagrams are forwarded based solely on the destination addresses stated in datagram headers and the next hops stored in the forwarding information bases (FIB) of  relaying routers allows Internet datagrams to traverse loops, which wastes resources and leaves the Internet open to further attacks. 
We introduce PEAR (Provenance Enforcement through Addressing and Routing), a new approach for  addressing and forwarding of Internet datagrams 
that enables anonymous forwarding of Internet datagrams, eliminates many of the existing DDoS attacks on the IP Internet, and prevents Internet datagrams from looping, even in the presence of routing-table loops. 

\end{abstract}

%----------------------------------------------------------------------------------------

\section{Introduction}
 
One of the biggest  challenges facing the  future of the  Internet is 
that  its  vulnerabilities to DoS and DDoS attacks are {\em inherent} in the algorithms used in the  Internet to: (a) assign addresses to hosts, routers, and devices; 
(b) include source addresses in Internet datagrams; (c) map  addresses to routes; (d) bind names to locations in the Internet; and (e) forward Internet datagrams.

In theory, the goal of  assigning  Internet Protocol (IP) addresses to entities and including the source and destination IP address in each datagram  is to have  system-friendly identifiers  that:  state the origin and destination of Internet datagrams based on topological locations where content,  services, or devices  are made available; and  can be matched efficiently against stored information by routers and end systems.
However, an  IP address simply denotes the point of attachment of a host or router to a network with a given IP address range, without any topological information other than the aggregation of IP addresses.  Furthermore,  IP addresses are assigned to entities independently of the establishment of routes to services, content, devices, groups, or any entity in general.  As a result,  routing protocols (e.g., OSPF, BGP) and directory services (e.g., DNS) map  names used to denote entities (e.g., domain names) to  names that denote points of attachment to networks (IP addresses).  In addition to this,  the Internet  surrenders any control of the allocation of source IP addresses to Internet datagrams, because
the origin of an Internet datagram specifies the source address of the datagram  independently of any forwarding mechanism and end nodes are allowed to specify IP source addresses.

Because of the algorithms  used to assign  IP addresses  to entities and write  source addresses  into Internet datagrams,  the source address of an Internet datagram fails to convey its provenance correctly.  The recipient of an Internet datagram is unable to authenticate the claimed IP address of the source of the datagram based solely on  the basic  operation of  the forwarding plane of the IP Internet.
The receivers of Internet datagrams are forced to  use additional  mechanisms  and information to cope with the fact that a source address
need not denote the valid  provenance of an Internet datagram. Furthermore, these mechanisms are far more complex than the simple mechanism used by  sources of Internet datagrams  to state the origins of  datagrams.  In addition, 
IP addresses are globally unique and  assigned on a long-term basis, which makes it easier for  attackers to plan and mount attacks. This constitutes a major vulnerability to DDoS attacks in the current Internet architecture.

In addition to the above,  a router forwards an Internet datagram to its next hop based solely on the destination address stated in the datagram and the next hop listed in its forwarding information base (FIB). This is a problem in the presence of routing-table loops, because it is possible for  Internet datagrams to traverse loops. The only approach used today is to include a time-to-live (TTL) field in the datagram header  that is decremented at each hop of the path traversed by the datagram, and to drop an Internet datagram after the TTL value reaches zero.

The contribution of this paper is to present a set of algorithms that we call 
PEAR (Provenance Enforcement through Addressing and Routing), and which 
prevent Internet datagrams from traversing forwarding loops, makes the identity of the origin of an Internet datagram anonymous to the rest of the Internet, and enforces the provenance of an Internet datagram.

Section \ref{sec-prior} summarizes current defenses against DDoS flooding attacks.
The main objective of this review of prior work is to point out that  defending against  large DDoS flooding attacks is virtually impossible without changing the basic algorithms used  for the allocation of addresses to Internet datagrams, the mapping of addresses to routes and connections, and the protection of information carried in Internet datagrams.  
Currently, attackers spend far less energy and time mounting attacks than their targets  spend defending against them. 

Section \ref{sec-loop} introduces a simple approach to ensure that Internet datagrams never loop, even when routing tables contain  long-term or short-term routing-table loops. The approach operates by having the FIB entry for an address prefix state the next hop {\em and} the 
hop-count distance  to the prefix, and 
by using the TTL filed of a datagram to enforce an ordering constraint ensuring that a router can forward a datagram only to a next hop that is strictly closer to the intended destination.

Section \ref{sec-anon} introduces a receiver-initiated address allocation algorithm,  a simple address swapping function, and an on-demand routing algorithm operating in the data plane, which  together ensure that the origins of Internet datagrams remain anonymous to any  routers processing the datagram, and that anonymous sources of datagrams can receive traffic from public destinations over the reverse paths traveled by their anonymous datagrams.

\section{Current Defenses against DDoS Attacks}
\label{sec-prior}

The methods used to launch DDoS attacks today consist of: (a) sending malformed packets to the victims to confuse protocols or applications; (b) disrupting the connectivity of legitimate users by exhausting bandwidth,  router processing capacity, or 
network resources; and (c) disrupting services to legitimate users by exhausting the resources of servers (sockets, CPU, memory, or I/O bandwidth).   
We   address DDoS flooding attacks aimed at disrupting the connectivity and services offered to legitimate users.
Four types  of defenses against these attacks have been proposed to date \cite{abliz, spoofing, mirkovic, peng, taghavi}: Attack prevention, which aims at stopping attacks before they reach 
their targets; 
attack detection, which attempts to  identify the existence of attacks when they occur; attack source identification, which tries to locate the source of the attack independently of the information contained in packets used in the attack; and attack reaction, which  aims at eliminating or minimizing the impact of attacks.

The attack prevention approaches proposed to date focus on routers filtering IP datagrams with spoofed source IP addresses (e.g., \cite{rpf95, ferguson, jin03, save, dward, dpf}), or routers {\em adding} provenance information to IP datagrams  
The limitations of existing packet-filtering approaches 
are  that: (a) existing filters provide only coarse-grained descriptions of valid source IP addresses; (b) filters are vulnerable to asymmetric routing and the dynamics of routing protocols; (c) it can be difficult to determine which source IP addresses are valid because of the complexity of the  network topology; and (d) the exiting approaches used to  update filters either incur substantial signaling overhead or are very slow to update filters with new data.
Prior approaches based on adding provenance information require too much effort 
by the attacked network, because they involve 
the use and dissemination of secret keys to mark IP datagrams,  and some even require additional headers.

Several  defense approaches have focused on 
detecting specific DoS and DDoS attacks (e.g., \cite{tops, blaz01, multops, peng04, wang02}) or anomaly detection (e.g., \cite{denning, komp, cdis}). These approaches monitor the behavior of specific protocols or the network in an attempt to detect flow anomalies. 
The key limitation of  defense approaches based on detection is that they must rely on a number of assumptions regarding the behavior of legitimate users, and attackers can adopt countermeasures to evade detection.
Some DoS attacks can be detected, given that only a few computer systems are attackers and compromised system must behave differently than benign users to exhaust the resources of their targets. However, the problem is 
far more difficult  for DDoS attacks, which involve many compromised hosts that can mimic legitimate users and need not  change the normal pattern of protocol traffic to be effective. 

Prior approaches for the identification of  attack sources   have focused on tracing the origins of attacks by explicit signaling or  marking datagrams with the paths they traverse (e.g., \cite{burch, savage, snoeren, song, stone, pi, stackpi}). Some path marking techniques have also been combined with filtering. The main limitations of these approaches include that: it may be difficult to infer the attack paths in large DDoS attacks; some approaches can consume considerable storage, processing, and communication overhead; some path markings are not entirely unique; and tracebacks and markings become useful only after the attacks have consumed network resources and have reached their targets.
Similarly,  prior approaches aimed at filtering Internet datagrams with spoofed IP addresses are only partially effective, because they change router behavior to enact filtering without changing the way in which IP source addresses are assigned to Internet datagrams \cite{rpf95, ferguson, jin03, save, dward, dpf},  or the fact that datagram forwarding is independent of the distances to address prefixes stated in routing tables. 

Approaches that introduce additional  information to denote the provenance of a datagram are  difficult to implement and cannot be deployed at Internet scale, because they require public key infrastructure (PKI) support.  For example, HIP \cite{hip} requires a  PKI that is globally deployed to prevent attackers from simply minting unlimited numbers of host identifiers used in HIP.  AIP \cite{aip} on the other hand assumes a flat addressing space that cannot be applied at Internet scale, is vulnerable to malicious hosts creating unlimited numbers of EIDs, and does not offer an efficient way to recovering from compromised private keys corresponding to the AIP addresses of hosts or accountability domains. Similar limitations exist in approaches that add path information (e.g., \cite{spm, passport}).

Attack-reaction approaches seek to minimize the damage caused by attacks by protecting bottleneck resources \cite{millen}. The mechanisms that have been proposed include reducing the state needed 
to execute specific communication protocols (TCP in particular \cite{bern96, eddy, schuba}); manage resources, shape traffic, and increase the capacity of servers \cite{karg, spat99}; and 
secure the communication between confirmed users and servers \cite{sos, tupa03}.
Attempting to reduce connection state in TCP is a good objective in all TCP implementation modifications; however, they have not caught on because of the inconsistencies they introduce in  the establishment and termination of connections.  Approaches  that attempt to manage bottleneck resources are not effective, because such resources are shared fairly by DDoS traffic, leaving limited resources for valid users. On the other hand, prior approaches that attempt to secure the communication between valid users and resources are overly complex, because they  build new mechanisms to enforce provenance and obfuscation on top of existing Internet methods for naming, addressing and routing.

\section{Eliminating Forwarding Loops in The Internet}
\label{sec-loop} 

The Internet datagram forwarding algorithm is  based on
FIB entries that simply  state the next hops for  IP address prefixes. 
To forward a datagram intended to destination address $d$, an IP router looks up its FIB to obtain the best match for $d$ among the entries listed in its FIB for known IP address prefixes and decrements the TTL of the datagram. 

Two approaches are currently used to cope with  the  occurrence of forwarding loops 
in Internet forwarding. In the forwarding plane,  the time-to-live (TTL) field of an Internet datagram is used to discard a datagram after it circulates a forwarding  loop too many times. In the control plane, routing protocols (e.g., OSPF and EIGRP) are used to reduce or eliminate the existence of routing loops. However, even if a loop-free routing protocol is used in the control plane, a datagram may still circulate a forwarding loop while routing tables are inconsistent among routers.

We propose using FIBs  to ensure that  forwarding decisions in the data plane are consistent with  the routing information maintained by the routing protocol operating in the control plane.  
The FIB entry stored at router $i$  for address prefix $d^*$ states the  minimum-hop distance $H^i_{d^*}$ to the prefix in addition to the next hop $n$ to the prefix. We assume for convenience that a distance stated in a FIB is a minimum-hop count to a destination.
We denote by $P^k[s^k, d, T^k, ID^o](p)$ 
an Internet datagram sent  by router $k$ with a header that contains a source address of local scope ($s^k$), a destination address of global scope ($d$),   a TTL value ($T^k$), and 
an origin ID  ($ID^o$),  plus  payload data $p$. Router $i$ uses the following rule  to forward such a datagram using its FIB within a network.
\\

\noindent
{\bf  TTL-based FIB Rule (TFR):}  \\
Router $i$  accepts  to forward $P^k[s^k, d, T^k, ID^o](p)$  from router $k$ towards the best-match prefix $d^*$  for $d$ if
$T^k > H^i_{d^*}$.

If router $i$ accepts $P^k[s^k, d, T^k, ID^o](p)$, it sets  $T^i = H^i_{d^*}$
and forwards datagram
$P^i[s^i, d, T^i, ID^o ](p)$ to its next hop towards prefix $d^*$.  
A router simply drops a datagram intended for a destination address of global scope  with a TTL value that does not satisfy TFR.
$\blacksquare$

TFR consists of imposing an ordering constraint  on the traditional Internet datagram forwarding algorithm based on FIB entries, and making the TTL value of the datagram equal to the distance stored in the FIB for the intended destination, rather than simply decrementing its value.
The following theorem proves that TFR eliminates IP forwarding loops.
\\

%\vspace{0.05in}
\begin{theorem}
\label{theo}
No Interest  can traverse a forwarding loop in an IP network in
which TFR is used to forward datagrams.
\end{theorem}

 \begin{proof}
Consider a network in which TFR is used and  assume for the sake of 
contradiction  that  routers in a forwarding  loop  $L$ of $h$ hops  $\{ v_1 , $ $v_2 , ..., $ $v_h , v_1 \}$  forward a datagram for destination $d$  along $L$ with no router in $L$  detecting that  the datagram  has traversed  loop $L$. 

Given that  $L$ exists by assumption, router $v_k \in L$ must forward  
$P^{v_k}[s^{v_k}, d, T^{v_k}, ID^o](p)$ to router $v_{k+1}$ $ \in L$ for $1 \leq k \leq h - 1$, and router $v_h \in L$ must forward 
$P^{v_h}[s^{v_h}, d, T^{v_h}, ID^o](p)$ to router $v_{1} \in L$. 

According to TFR, if router $v_k$  ($1 < k \leq h$)
forwards 
$P^{v_k}[s^{v_k}, d, T^{v_k}, ID^o](p)$  to router $v_{k+1}$ as a result of receiving  
$P^{v_{k - 1}}[s^{v_{k - 1}}, d, T^{v_{k - 1}}, ID^o](p)$  from router $v_{k-1}$, then  it must be true that  
$T^{v_{k - 1}} > H_{d^*}^{v_k}$, where $d^*$ is the address prefix that is the best match for destination $d$.
Similarly, if router $v_1$ 
forwards $P^{v_{1}}[s^{v_{1}}, d, T^{v_{1}}, ID^o](p)$  to router $v_{2}$ as a result of receiving  $P^{v_{h}}[s^{v_{h}}, d, T^{v_{h}}, ID^o](p)$ from router $v_{h}$, then  it must be true that  
$T^{v_{h}} > H_{d^*}^{v_1}$.
However, these results constitute a contradiction, because they imply that  $H_{d^*}^{v_{k}} > H_{d^*}^{v_{k}}$ for $1 \leq k \leq h $. Therefore, the theorem is true.  
\end{proof}

%  \vspace{0.05in}
Theorem 1 is  independent of whether the network is static or dynamic, or whether routers use single-path routing of multipath routing.
The ordering constraint  of TFR is essentially the same loop-free condition first introduced in DUAL \cite{dual}.
The difference between the way in which the ordering constraint is used in TFR and  in DUAL
is that TFR establishes distance-based ordering in the data plane  to forward datagrams based on existing FIB entries, while  DUAL  establishes distance-based ordering  in the control plane to build  FIB entries that are  then used to  determine how to forward datagrams.

As stated, TFR is only applicable within an autonomous system (AS)  in which  the same routing protocol is used to obtain minimum-hop distances to destinations. However, applying the approach used in TFR across 
autonomous systems  is relatively straightforward. A datagram needs to carry  two hop counts, one that states the AS hop count to the destination AS, and another one stating the distance to the intended destination within the same AS or the gateway connecting to the next AS along the path to the destination AS.

\section{Anonymous Datagram Forwarding  }
\label{sec-anon}

Since the introduction of datagram packet switching by Baran \cite{baran}, the identifiers used to denote the destinations {\em and} sources of datagrams have had global scope. Today, the FIBs maintained by Internet routers list entries that state  the next hop to each known IP address range of global scope.

However, five important observations can be made to argue that datagram forwarding does not have to be limited to Baran's original  design.  First, the purpose of having addresses in datagram headers is to enable   hop-by-hop datagram forwarding based on fast lookups of  destination-based routing tables, which does not require addresses to have  global, long-term meaning. Second, assigning IP addresses to hosts on a long-term basis  as if they were names is not necessary and enables attackers to take advantage of the quasi-static nature of such identifiers.  Third, there is no technical reason for the origin of an IP datagram to be the entity that assigns the source IP address of the datagram.  Fourth,  having the destination of an IP datagram ascertain the provenance of a datagram without any assistance from the routing infrastructure imposes too much effort on the destination and is done {\em after} datagrams with spoofed addresses have wasted network resources. Fifth, the packet-filtering schemes proposed to date to address DDoS attacks do not take full advantage of  the distance information maintained in routing tables. 

Our approach consists of using  IP addresses of local scope to denote 
anonymous sources of Internet datagrams, and introducing an on-demand routing algorithm to maintain routes to such addresses. Similar to the approach in \cite{handley}, we advocate Internet datagram forwarding over symmetric paths in order to support the forwarding of Internet datagrams using IP addresses with local scope.
The rest of this section describes the mechanisms for supporting anonymous datagram forwarding  in the Internet using a simple example. A formal description of the algorithms is omitted for brevity.

\subsection{Receiver-Initiated Source Address Assignment}

Each router announces  a continuous  interval of IP addresses ({\em local interval}) that the router considers valid local-scope addresses used to denote sources or destinations  of Internet datagrams. 

Router $i$   maintains a local-interval set table ($LIST^i$) that lists  the local interval announced by each neighbor router $k$ ($LI^i (k)$) to router $i$, and the 
local interval  announced by  router $i$ to its neighbors ($LI^i (i)$).
All local intervals are of equal length $|LI|$, and hence a   local interval is uniquely defined by the IP address at the start of the interval. 
The start of local interval $LI^i (k)$ is denoted by $LI^i (k)[s]$.
Given this,  router $i$ can easily map an IP address of local scope $x \in LI^i(i)$ to a corresponding IP address of local scope $y$  that is acceptable to neighbor $k$ with the following bijection, where $\epsilon$ is a constant known only to router $i$: 

% \vspace{-0.05in}
{\small
\begin{equation}
\label{bijection}
y  \equiv    \epsilon + x - LI^i(i)[s] + LI^i(k)[s]  ~\mod~ |LI|
\end{equation}
}

 % \vspace{-0.2in}
\subsection{Forwarding Information Stored and Exchanged}

PEAR is transparent to end hosts.  Clients and servers simply use IP datagrams without modification. 
% ; a datagram sent between an end host and its attached router is denoted by $D[x, y]$ where $x$ and $y$ represent the source and destination address in the datagram, respectively. 
The addresses specified in a datagram sent between a host and router   may be of local or global scope. A host that originates datagrams must  use an IP address taken from the local-interval set used by the adjacent router.

To forward Internet datagrams from an anonymous datagram source to a destination with an IP address $d$ that has  global scope, 
PEAR  
uses three addresses to forward a datagram. A source IP address of local scope
denotes the originating router of the datagram at each hop of the path towards the destination.  An   IP address of global scope denotes the intended destination. In addition,   an  IP address of local scope, which we call  the origin ID, is used to denote the host that originated the datagram.
%  in order to reduce the amount of forwarding state  maintained by relay routers. 

The TTL filed of a datagram is used to enforce 
TFR at each router to forward a datagram towards the best prefix match $d^*$ for destination address $d$ of global scope. A  router that receives a datagram from a host sets the TTL field of the datagram to equal its hop-count distance to the destination.

Because the source IP address and the origin ID stated in a datagram have local scope, routers along the path from an anonymous source to a destination with a global-scope address must establish forwarding state pointing back towards the anonymous source in order for datagrams from $d$ to the anonymous source to be forwarded.
Additional forwarding tables are needed to accomplish this.

Each router maintains  a {\em Hop-Specific Routing Table} ({\bf HRT}) to keep track of the next hops towards destinations denoted 
by  addresses that have hop-specific scope.
Router $i$ maintains $HRT^i$, which  is indexed using hop-specific IP addresses taken from $LI^i (i)$.   
An  entry in $HRT^i$ states  a hop-specific IP address (HIP) used to denote a destination   ($HIP(HRT^i) \in LI^i (i)$), a next hop to that destination,
$(n(HRT^i) \in$ $ N^i$), and a hop-specific address mapping used to handle collisions of hop-specific IP addresses $(map(HRT^i) \in LI^i (i)$).

$HRT^i [HIP, n, map]$  denotes a given  entry in $HRT^i$.  A HIP used as the source IP address in a datagram forwarded by router $i$
is denoted by $SHIP^D(i)$. Similarly, 
$DHIP^D(i)$ and $OHIP^D(i)$ are used to denote 
the destination IP address and the origin IP of a datagram forwarded by router $i$, respectively.The origin IDs    {\em are not} stored in HRT entries. Each router uses the same bijection of Eq. (1)  to map the origin ID stated in a datagram it receives to the origin ID of the datagram it forwards towards its destination.

Each router also maintains a {\em Destination Routing Table} ({\bf DRT}) to reduce the amount of  forwarding state needed in relaying routers. 
Router $i$ maintains $DRT^i$, which  indexed using origin IDs 
received in  datagrams and lists entries corresponding to sources of datagrams that have indicated a destination address that is directly attached to router $i$.
An  entry in $DRT^i$ states  an origin ID  ($(O(DRT^i)$) and a  hop-specific IP address ($(HIP(DRT^i)$), both of which are in the local interval defined for router $i$.

\begin{figure}[h]
\begin{centering}
    \mbox{
    \subfigure{
  \scalebox{.2}{
    \includegraphics{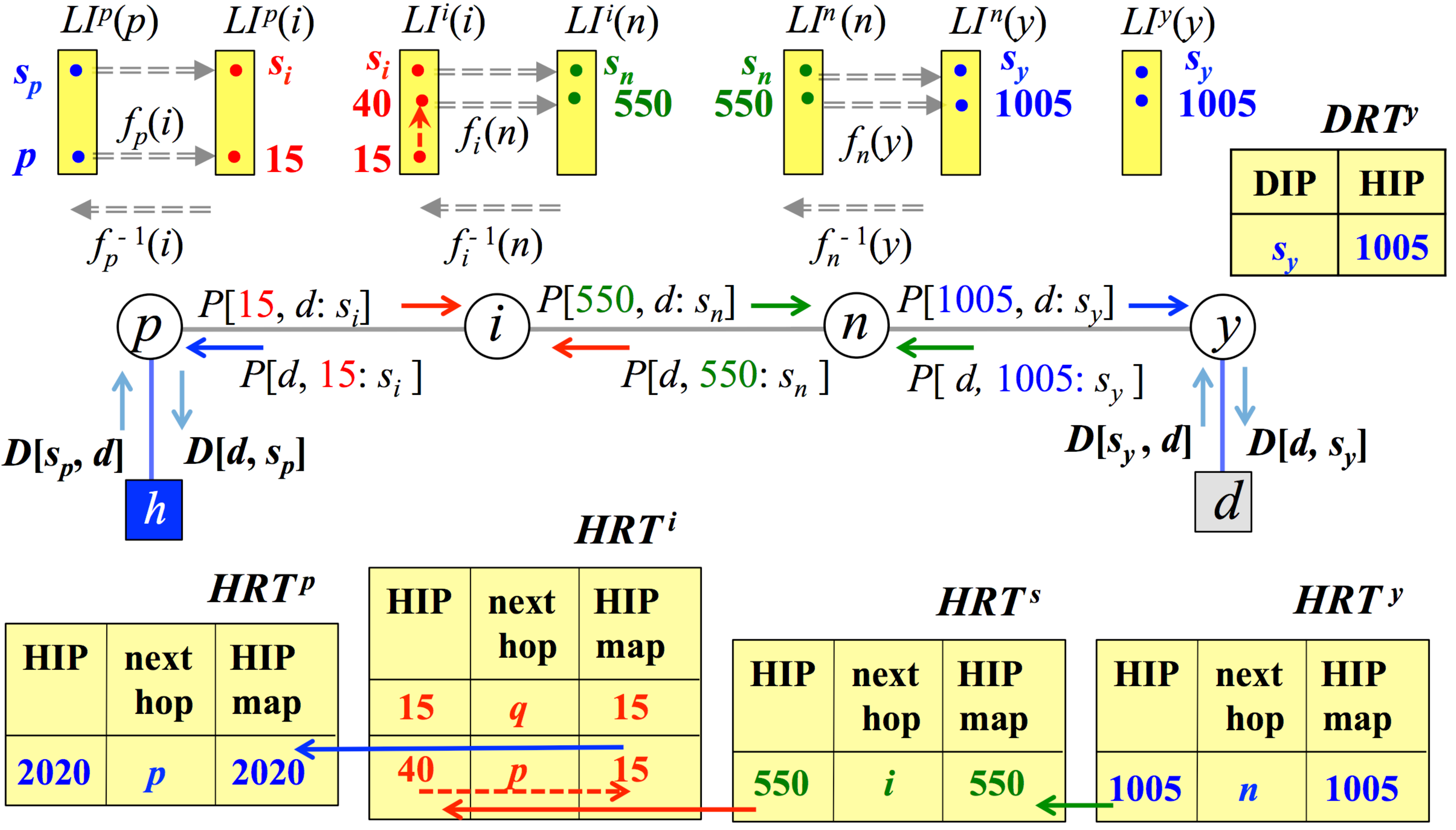}}}
     }
   \caption{Datagram forwarding 
   using  IP  address swapping with PEAR
   }
   \label{pear1}
\end{centering} 
\end{figure}  

%  \vspace{-0.1in}
\subsection{Forwarding of Datagrams from Anonymous Sources }

Routers that forward a datagram from an anonymous source to a destination address of global scope use their FIBs to determine the next hop and establish forwarding state in their HRTs as the datagram is forwarded to allow datagrams to flow back to anonymous sources.

Figure~\ref{pear1} shows an example of the  forwarding state established 
at relay routers between an anonymous source and a well-known destination to enable datagrams to be sent back to IP addresses of local scope.  The example in the figure 
shows a client $h$  initiating communication with a server assigned an address $d$ that   has network-wide scope and is part of an address prefix listed in the forwarding information bases (FIB)  maintained by  routers. It is assumed that  the client has obtained the IP address of the intended destination $d$ using the DNS or other means.  The IP address assigned to host $h$ is $s_p$ and has only local scope at router $p$, i.e., other routers do not associate $s_p$ with the same host $h$.
In the example, host $h$ must use  IP address $s_p$ as its source address.
Router $p$ accepts datagrams from host $h$ only if the IP source address in those datagrams is $s_p$. 
The TTL field is not shown in the figure, but TFR is assumed to be applied by each relay router.

Starting with the router that receives a datagram from an attached host (router $p$ in the example), routers establish forwarding state 
using a source IP addresses of local scope,  an  origin ID  that is also an IP address of local scope, and the destination address.
In the figure, a datagram forwarded  between two routers towards the destination with global-scope address $d$ is denoted by
$P[x, d: s_x]$, where $x$ is a source IP address of local scope, $d$ is the global-scope  address, and $s_x$ is an origin ID.
The source IP address and origin ID of a datagram  being forwarded by router $r$ to router $n$ along the path  from host $h$ to address $d$  ($P_{hd}$)  are taken from the local interval assigned by $n$ to $r$ ($LI^r(n)$) using the bijection stated in  Eq. (1).

In the example of Figure~\ref{pear1}, router $i$ has an exiting entry $HRT^i [15, q, 15]$  when it receives  datagram $P[15, d: s_i ]$ from router $p \not= q$. The destination IP address in the datagram has global scope and hence is forwarded  based on its FIB. However, the source IP address in the datagram received from $p$  collides  with  entry $HRT^i [15, q, 15]$.  Accordingly, router $i$ selects $HIP = 40$, which is not used in any 
existing HRT entry; 
creates  entry $HRT^i[ 40, p, 15]$,  sets  $SHIP^D(i) = f_i(n)[40] = 550$ and  $OHIP^D(i) = f_i(n)[s_i] = s_n$; and forwards
datagram $P[ 550, d: s_n ]$ to router $n$.  

When router $y$ receives  $P[ 1005, d: s_y ]$ from router $n$, it determines that the host with address $y$ is locally attached.
Accordingly, it creates an entry in $DRT^i$ with the tuple  $( s_y,  1005  )$ and passes datagram $D[s_y, d]$ to the host with IP address $d$.

\subsection{Forwarding of Datagrams to Anonymous Sources }

Routers use their HRTs to forward datagrams to destinations denoted with hop-specific addresses. Given that TFR is used when the entries in HRTs are established,  the paths implied by HRT entries are loop-free. 

In Figure 1, when router $i$ receives datagram $P[d, 550: s_i]$ from router $n$, the fact that the destination address in the datagram is in its local interval $LI^i(i)$ instructs router $i$ to use $HRT^i$ for forwarding rather than its $FIB^i$. To do so, router $i$  computes the inverse function  $f_i^{-1}(n)[550] = 40$.
Using $HIP = 40$ as the key in $HRT^i$, router $i$ obtains the next hop $p$ and 
sets $DHIP^D(i) = 15$. Router $i$ also computes $f_i^{-1}(n)[s_n] = s_i$, sets $OHIP^D(i) = s_i$, and forwards  forwards  $P[d, 15: s_i]$ to router $p$.
In turn, router $p$ uses $f_i^{-1}(p)$ to obtain the values of the destination address and origin ID for the datagram it should forward. However, given that 
$p = f_i^{-1}(p)[15]  \in L^p(p)$, router $p$ obtains the local host that should receive the datagram. It computes $f_i^{-1}(p)[s_i] = s$ and forwards
$D[d, s]$ to the host with IP address $s$.

\subsection{Enforcing Anonymity and Provenance of Datagrams  }

It is easy to show that PEAR enforces anonymity of datagrams, in the sense that no intruder can determine the origin of a datagram simply by monitoring traffic over a link and reading the headers of datagrams.  The reason an Internet datagram divulges the identity of its source is that source addresses have global scope and hence any relay router or intruder receiving the datagram can determine its origin. 

Reducing the scope of a  source address to the specific hop where the datagram is being forwarded still allows intruders and relay routers to infer the identity of the source in small networks. However, forcing each relay to use an address from an address space provided by the next hop eliminates the ability of an intruder or the receiving router to infer the true identity of the datagram origin. This is the case  even in small networks larger than two routers, unless the topology of the network is such that any traffic sent from a given router must be originated by the router (e.g., a leaf router).

Even tough PEAR uses IP addresses of hop-specific scope to provide anonymity, routers can enforce correct provenance of datagrams by eliminating the ability of hosts or routers to inject arbitrary source IP addresses in datagrams being forwarded among trusted routers implementing PEAR. 

An ingress router accepts a datagram from an attached host only if it has a source address deemed valid for that host. Furthermore, the source IP address and origin ID it uses in the datagram it forwards to the next hop must both be in the address space provided by the next-hop router, and the same is true at every hop along the path to the destination. Therefore, it follows that any malicious host or team of routers attempting to inject datagrams with spoofed IP source addresses can be identified by the first trusted router that does not receive datagrams with correct source addresses. Furthermore,  malicious datagrams using valid source addresses can be traced back to the first untrusted router injecting the traffic, given the inductive nature of the receiver-initiated source address assignment and the simple bijection used at each hop to swap IP addresses to and from anonymous destinations.

% \vspace{-0.1in}
\section{Conclusions and Future Work}

We introduced the first approach for  Internet datagram forwarding that eliminates forwarding loops even if routing tables contain routing-table loops. Based on the ability to eliminate forwarding loops, we also presented an approach that enables the origins of Internet datagrams to remain anonymous while enforcing the correct provenance of datagrams.  As such, the proposed forwarding scheme can be applied to  eliminate many  of the existing DDoS attacks on the IP Internet,
which are enabled by the inability of Internet datagrams to enforce correct provenance.

Additional work is needed to formally describe the specific DDoS countermeasures enabled by PEAR. Equally important, TFR can be applied to inter-domain routing by utilizing two hop counts, one that takes effect across autonomous systems and another one that is used within an autonomous system. However, such cluster-based version of TFR has not been proven to be correct as is the case of TFR \cite{ifip2015, dual} and this is an obvious next step.

Lastly, even though we have advocated symmetric paths as additional protection against DDoS attacks \cite{handley}, the Internet has many asymmetric paths, and the proposed scheme must be extrapolated to the case in which datagrams from $s$ to $d$ follow a different path than datagrams from $d$ to $s$.

\section{Acknowledgments}

This work was supported in part by the Jack Baskin Chair of Computer Engineering at the University of California at Santa Cruz.

% \vspace{-0.1in}

\end{document}